\documentclass[a4paper,12pt]{article}
\usepackage{amssymb,amsmath,amsthm}
\usepackage{graphics,graphicx,color}
\usepackage{srcltx}
\usepackage{fancyhdr}
\DeclareGraphicsExtensions{.eps} \textwidth=169mm
\definecolor{darkred}{rgb}{0.9,0.1,0.1}

\newtheorem{theorem}{Theorem}[section]
\newtheorem{lemma}{Lemma}[section]

\newtheorem{remark}{Remark}[section]
{\rm}
\definecolor{darkred}{rgb}{0.9,0.1,0.1}

\begin{document}

\title{Resolvent  bounds for  jump generators and ground state asymptotics for nonlocal Schr\"{o}dinger operators
\thanks{The work has been partially supported by SFB 701 (Universitat Bielefeld). The research of A.Piatnitski
and E.Zhizhina has been supported by the Russian Science Foundation (project No. 14-50-00150). The research of
S.Molchanov
has been partially financially supported by the National Science Foundation Grant "Asymptotic and Spectral Analysis of Applied non-self-adjoint problems"
} }

%\footnotemark[3]\footnotetext[3]
%\let\thefootnote\relax\footnote
%\footnotemark[1]\footnotetext[1]
%\renewcommand{\thefootnote}{\arabic{footnote}}
%\let\thefootnote\relax\footnote

\author{
Yuri Kondratiev$^1$
\and 
Stanislav Molchanov$^{2,3}$ 
\and 
Andrey Piatnitski$^{4,5}$ 
 \and Elena
Zhizhina$^6$ }

\date{\begin{minipage}{0,8\textwidth}
{$^1$\small{Fakult\:at fur
 Mathematik, Universitat Bielefeld, 33615 Bielefeld, Germany}} \\[-2mm]
${ }$\ {{\small
(\tt kondrat@math.uni-bielefeld.de)}}\\[2mm]
$^2${\small Department of Mathematics and Statistics, \\[-1.5mm]
${ }$\ University of North Carolina at Charlotte, USA\\[-2mm]
${ }$\ (\tt smolchan@uncc.edu)}\\[2mm]
$^3${\small National Research University "Higher School of Ecinomics", Moscow, Russia}\\[2mm]
{$^{4}${\small The Arctic University of Norway, Tromso, Norway (\tt andrey@sci.lebedev.ru)}}\\[-1mm]
{$^{5}${\small Institute for Information Transmission
Problems, Moscow, Russia }}\\[2mm]
{$^{6}${\small Institute for Information Transmission Problems,
Moscow, Russia\\[-2mm] ${ }$\ (\tt ejj@iitp.ru).}}
\end{minipage}
}

%\date{}

\maketitle



\begin{abstract}
The paper deals with jump generators with a convolution kernel. Assuming that the kernel
decays either exponentially or polynomially we prove a number of lower and upper bounds
for the resolvent of such operators. We consider two applications of these results.
First we obtain pointwise estimates for principal eigenfunction of jump generators perturbed
by a compactly supported potential (so-called nonlocal Schr\"odinger operators).  Then we consider the Cauchy problem for the corresponding inhomogeneous evolution equations and study the behaviour of its solutions.
\end{abstract}

\section{Introduction}

This paper deals with non-local operators with an integrable  convolution kernel.
%, see (\ref{3}) for their definition.
The properties of such operators depend crucially on the behaviour of the convolution kernel at infinity.  We consider different cases covering both polynomial and exponential rates of decay of its tails.
 Our main aim is to investigate the behavior at infinity of the resolvent kernel of these operators. 
 We prove a number of lower and upper bounds for the said resolvent kernel and,  with the help of these results,  deduce  pointwise bounds for the principal eigenfunction  of the operator obtained from the original convolution type operator by adding a compactly supported potential.
  Another application concerns   the inhomogeneous Cauchy problems for convolution type operators. 
 %Here our goal is to establish sharp lower and upper for  the principal eigenfunction
 %(ground state).

The non-local operators considered here are generators of Markov jump processes with the jump distribution defined by the convolution kernel. The analysis of the behavior of the processes can be performed in terms of the resolvent of its generators.
 On the other hand,  such operators appear in the kinetic description of birth-and-death Markov dynamics of populations in spatial ecology, see  e.g. \cite{FKK} and the literature therein.
In this setting the tail of the kernel can be thought of as the range of dispersion  of newborn individuals. In many applications this kernel might have heavy tails.

The problem of existence of discrete spectrum and the principal eigenfunction for the perturbed operator  has been studied in recent papers \cite{BCV, C10, KMPZh, KMV}.   This operator may be considered as a non-local version of  the Schr\"{o}dinger operator in which the Laplacian is changed to a convolution type operator. It is interesting to observe that the ground state problem for such non-local Schr\"{o}dinger operator
is rather different comparing with that in the classical local case.  We will come to this point later on.

%Notice that in contrast with \cite{CCR} where only the asymptotic estimates for resolvent have been
%obtained, we provide the bounds valid for all $x \in \mathbb R^d$.
To our best knowledge the only work in the existing literature where the behaviour of the resolvent kernel for non-local operators has been studied is \cite{KMV}.
In the mentioned work  the authors consider the case of non-local operators with the convolution kernel decaying super exponentially. The fast decay of the kernel allows one to use the techniques of analytic functions.

In the present paper we consider both polynomially decaying and exponentially decaying convolution kernels. In the former case under some natural conditions we show that the resolvent kernel has the same polynomial decay, see Theorem \ref{t_1}. In the later case the resolvent kernel decays exponentially, see Theorems \ref{t_2}, \ref{l2}. These results imply in particular that the ground state of the perturbed operator shows the same rate of decay at infinity as the dispersal kernel of the non-perturbed operator.

In both cases we also analyze the limit behavior of the mentioned kernels as the spectral parameter tends to zero. This asymptotics play a crucial role when studying the spectral problem for the perturbed operator with "small enough" potential. Then the corresponding principal eigenvalue is small, and the behavior of the principal eigenfunction at infinity can be expressed in terms of the resolvent kernel with a small value of the spectral parameter, see Theorem \ref{T1_1}.

We also consider in this work the behaviour of solutions to the Cauchy problem for a non-local heat equations with a stationary source term. In particular, we provide lower and upper bounds for the solutions and prove the convergence  to the stationary solution. Previously a number of qualitative results  for non-local heat equations have been obtained in a number of works, see \cite{ AMRT, CCR}
and references therein.

The paper is organized as follows. Section 2 deals with the behaviour of the resolvent of integral operators with convolution kernels. In subsection 2.1 we study the case of polynomially decaying kernels
while subsection 2.2 is devoted to the case of exponentially decaying tails.

These results are then used in Section 3 to study the properties of the principal eigenfunction
of the perturbed operator  and the large time behaviour of solutions to the Cauchy problem for the corresponding non-homogeneous evolution equation.

\section{Bounds for resolvent kernel}

We consider the operator
\begin{equation}\label{3}
(L_0 u)(x) \, = \, \int_{\mathbb{R}^d} a(x-y) (u(y) - u(x)) dy.
\end{equation}
Throughout this paper we assume the following properties for the dispersal kernel: $a(\cdot)  \in C_b(\mathbb R^d) \cap L^1({\mathbb R}^d) $ is a nonnegative bounded even continuous function of unit mass, i.e.
\begin{equation}\label{aL1}
\int_{\mathbb R^d} a(x) dx = 1.
\end{equation}
Consequently
\begin{equation}\label{a2}
a(\cdot) \in  L^2 ({\mathbb R}^d) \;\;  \mbox{ and  its Fourier transform } \;  \tilde a(\cdot) = \int_{\mathbb R^d} e^{-i(\cdot, x)} a(x) dx \in L^2 ({\mathbb R}^d) \cap C_b(\mathbb R^d).
\end{equation}

For the further analysis it is convenient to rewrite operator $L_0$ in (\ref{3}) as follows:
\begin{equation}\label{3bisbis}
L_0 u(x) \, = \, \mathcal{S}_a u(x) - u(x), \quad \mathcal{S}_a u(x) = \int_{\mathbb{R}^d} a(x-y) u(y)  dy.
\end{equation}

In this section we study the behavior of the kernel of the resolvent for the operator $L_0$ under the condition that the function $a(x)$ decays either polynomially or exponentially. The resolvent kernel admits the representation
\begin{equation}
\label{GUS}
\mathcal{R}_\lambda(x,y)=(1+\lambda)^{-1}\Big(\delta(x-y)+ G_\lambda(x-y)\Big), \quad \lambda \in (0, \infty),
\end{equation}
where $G_\lambda(x-y)$ is the kernel of the convolution operator
\begin{equation}\label{Gbis}
\sum\limits_{k=1}^\infty\frac{\mathcal{S}_a^k }{(1+\lambda)^k}.
\end{equation}
Denote by $a_k(x-y)$ the kernel of the operator $\mathcal{S}_a^k$, then
\begin{equation}\label{ak}
a_k(x)=a^{\ast k}(x)=\int\limits_{\mathbb R^d}\ldots\int\limits_{\mathbb R^d}a(x-y_1) a(y_1-y_2)\dots a(y_{k-2}-y_{k-1})a(y_{k-1})\,dy_1\ldots dy_{k-1},
\end{equation}
and
\begin{equation}\label{G}
G_\lambda(x) \ = \ \sum_{k=1}^\infty \frac{a_k(x)}{(1+\lambda)^k}
\end{equation}
with $a_k(x)$ defined by (\ref{ak}).

\subsection{Polynomial tail of  dispersal kernel}

In this section we deal with functions $a(x)$ that satisfy the following bounds
\begin{equation}\label{bou_a1}
c_-(1+|x|)^{-(d+\alpha)} \leq a(x) \leq c_+(1+|x|)^{-(d+\alpha)},
\end{equation}
with an arbitrary $\alpha>0$.

\begin{theorem}\label{t_1} There exist constants $0<\tilde c_-(\lambda)\leq \tilde c_+(\lambda)$, such that
\begin{equation}\label{main_ineq1}
\tilde c_-(d+|x|)^{-(d + \alpha)} \leq G_\lambda(x) \leq \tilde c_+(1+|x|)^{-(d + \alpha)},
\end{equation}
where $G_\lambda(x)$ is defined by (\ref{Gbis}).

Furthermore, $ \tilde c_+(\lambda)= O(\lambda^{-(2+d+\alpha)})$ as $\lambda$ goes to 0, and
$$
G_\lambda(x)\geq \frac{C_0}{\lambda}(1+|x|)^{-(d+\alpha)}
$$
for large enough $|x|$ with a constant $C_0$.

\end{theorem}

\begin{proof}

I. The upper bound. In $\mathbb R^{(k-1)d}$ introduce the sets
$$
\begin{array}{c}%\displaystyle
A_1=\big\{y\in\mathbb R^{(k-1)d}\,:\,|x-y_1|\geq \frac{|x|}k\big\},
\ A_2=\big\{y\in\mathbb R^{(k-1)d}\,:\,|y_1-y_2|\geq \frac{|x|}k\big\},\ldots,
\\[2mm]
A_k=\big\{y\in\mathbb R^{(k-1)d}\,:\,|y_{k-1}|\geq \frac{|x|}k\big\}.
\end{array}
$$
One can easily check that $\mathbb R^{(k-1)d}\subset\bigcup_{j=1}^k A_j$ and, therefore,
$$
a_k\leq \sum\limits_{j=1}^k\int_{A_j}a(x-y_1) a(y_1-y_2)\dots a(y_{k-2}-y_{k-1})a(y_{k-1})\,dy_1\ldots dy_{k-1}.
$$
Thus considering (\ref{bou_a1}) we obtain
$$
\begin{array}{c}
\displaystyle
\int_{A_1}a(x-y_1) a(y_1-y_2)\dots a(y_{k-2}-y_{k-1})a(y_{k-1})\,dy_1\ldots dy_{k-1}\\[5mm]
\displaystyle \leq
\frac{c_+}{(1+|k^{-1}x|)^{d+\alpha}}
\int_{A_1} a(y_1-y_2)\dots a(y_{k-2}-y_{k-1})a(y_{k-1})\,dy_1\ldots dy_{k-1}\\[5mm]
\displaystyle \leq
\frac{c_+}{(1+|k^{-1}x|)^{d+\alpha}}
\int_{\mathbb R^{k-1}} a(y_1-y_2)\dots a(y_{k-2}-y_{k-1})a(y_{k-1})\,dy_1\ldots dy_{k-1}=
\frac{c_+}{(1+|k^{-1}x|)^{d+\alpha}}.
\end{array}
$$
Similarly,
$$
\int_{A_j}a(x-y_1) a(y_1-y_2)\dots a(y_{k-2}-y_{k-1})a(y_{k-1})\,dy_1\ldots dy_{k-1}\leq
\frac{c_+}{(1+|k^{-1}x|)^{d+\alpha}}
$$
for $j=1,2,\ldots, k-1$. Finally,
\begin{equation}\label{akbis}
a_k(x)\leq\frac{c_+ k}{(1+|k^{-1}x|)^{d+\alpha}}
\end{equation}
Denote $\Lambda^2=(1+\lambda)$, then
$$
\Lambda^{-2k} a_k(x)\leq \frac{c_+ k\Lambda^{-k}}{(\Lambda^{k/(d+\alpha)}+
|\Lambda^{k/(d+\alpha)}k^{-1}x|)^{d+\alpha}}\leq
\frac{c_+ k\Lambda^{-k}}{(1+
|\gamma x|)^{d+\alpha}}
$$
with $\gamma=\min\limits_k \Lambda^{k/(d+\alpha)}k^{-1}$. Notice that $\gamma=\gamma(\lambda)>0$, and $\gamma = O(\lambda)$ for small $\lambda>0$. Summing up these inequalities in $k$ we conclude that the kernel $G_\lambda (x)$ is bounded from above by
$$
\frac{\tilde c_+ }{(1 + | x|)^{d+\alpha}}.
$$
This completes the proof of the upper bound in \eqref{main_ineq1}.

We proceed with the asymptotics of  $\tilde c_+ (\lambda)$ as $\lambda \to 0$. Using (\ref{akbis}) we get the following upper bound for the sum (\ref{G}):
$$
\sum_{k=1}^{\infty} \frac{c_+ k^{1+d+\alpha}(1+\lambda)^{-k}}{(k+|x|)^{d+\alpha}} \leq  \frac{c_+}{(1+|x|)^{d+\alpha}} \sum_{k=1}^{\infty} k^{1+d+\alpha}(1+\lambda)^{-k}.
$$
Consequently,
$$
\tilde c_+ (\lambda) = c_+ \sum_{k=1}^{\infty} \frac{k^{1+d+\alpha}}{(1+\lambda)^{k}} = O(\lambda^{-(2+d+\alpha)}) \quad \mbox{ as } \; \lambda \to 0.
$$
In order to justify the last relation we first restrict the summation to a finite range of $k$.
Namely,
$$
\sum_{k=1}^{\infty} \frac{k^{1+d+\alpha}}{(1+\lambda)^{k}}\geq \sum_{k=1}^{[\log 2/\lambda]} \frac{k^{1+d+\alpha}}{(1+\lambda)^{k}}\geq \frac12\sum_{k=1}^{[\log 2/\lambda]}k^{1+d+\alpha}
\geq C\lambda^{-(2+d+\alpha)}
$$
To obtain an upper bound we estimate the contribution of the terms related to $k\in J_m=\{k\in\mathbb Z^+\,:\, m\lambda^{-1}\log2\leq k\leq (m+1)\lambda^{-1}\log2\}$, $m\in\mathbb Z^+$:
$$
\sum_{k\in J_m} \frac{k^{1+d+\alpha}}{(1+\lambda)^{k}}\leq \sum_{k=J_m} \frac{k^{1+d+\alpha}}{2^{m}}\leq \frac{(m+1)^{2+d+\alpha}}{2^m}\lambda^{-(2+d+\alpha)}.
$$
Summing up in $m$ we arrive at the desired bound
$$
\sum_{k=0}^\infty \frac{k^{1+d+\alpha}}{(1+\lambda)^{k}}\leq C\lambda^{-(2+d+\alpha)}.
$$

\vskip5mm
II. The lower bound. The lower bound is quite straightforward, if we take in (\ref{GUS}) the first term of the sum.
To trace the dependence of the lower bound on $\lambda$ for small values of $\lambda>0$ we should take into account the higher order terms in (\ref{G}).

From \eqref{bou_a1} it follows that
$$
\int\limits_{|y|\geq r}a(y)\,dy\leq C r^{-\alpha}
$$
with a constant $C=C(c_+,\alpha,d)$ that only depends on $c_+, \ \alpha$ and $d$. Then
\begin{equation}\label{pc}
\int\limits_{|y|\leq r}a(y)\,dy\geq 1-C r^{-\alpha}.
\end{equation}

\begin{lemma}\label{l_1_bis}
For all $\lambda \in (0,1)$ and for all
\begin{equation}\label{L1}
|x|\ge \Big(\frac{1}{\lambda}\Big)^{\frac{\alpha+1}{\alpha}} m \quad \mbox{ with } \; m =
(2C)^{\frac{\alpha+1}\alpha},
\end{equation}
where $C$ is the same constant as in \eqref{pc}, we have
$$
G_\lambda(x)\geq \frac{C_0}{\lambda}(1+|x|)^{-(d+\alpha)},
$$
where $C_0$ is a positive constant that does not depend on $\lambda$.
\end{lemma}

\begin{proof}
After change of variables we get
$$
a_k(x)=\int\limits_{\mathbb R^d}\ldots\int\limits_{\mathbb R^d}a(z_1) a(z_2)\dots a(z_{k-1})
a(x-z_1-z_2-\ldots-z_{k-1})\,dz_1\ldots dz_{k-1}
$$
$$
\geq\int\limits_{B_r}a(z_1) a(z_2)\dots a(z_{k-1})
a(x-z_1-z_2-\ldots-z_{k-1})\,dz_1\ldots dz_{k-1},
$$
where $B_r=\{z\in \mathbb R^{(k-1)d}\,:\,|z_j|\leq r\}$. For $z\in B_r$ the following inequalitty is fulfilled:
$$
a(x-z_1-z_2-\ldots-z_{k-1})\geq c_-(1+|x|+kr)^{-(d+\alpha)}.
$$
Moreover,
$$
\int\limits_{B_r}a(z_1) a(z_2)\dots a(z_{k-1})\,dz\geq(1-C r^{-\alpha})^{k-1}
$$
with a constant $C$. This yields
$$
a_k(x)\geq c_-(1-C r^{-\alpha})^{k-1}(1+|x|+kr)^{-(d+\alpha)}.
$$
Notice that under our choice of the constant $m$ in (\ref{L1}) for all $\xi \geq m$ the inequality holds
\begin{equation}\label{m}
1- C \xi^{-\frac{\alpha}{\alpha+1}} \ge e^{-2C \xi^{-\frac{\alpha}{\alpha+1}}}.
\end{equation}
Then we consider all positive integer $k$ that satisfy the estimate
$$
k \leq k_0(|x|) = \left[ |x|^{\frac\alpha{\alpha+1}} \right]+1,
$$
and set
$$
r= r(|x|)=|x|^{\frac1{\alpha+1}}.
$$
Then using (\ref{m}) we have that
$$
1-C r^{-\alpha} \ge e^{-2 C r^{-\alpha}} \quad \mbox{ for all } \quad |x| \ge m,
$$
and for all $k \le k_0$ we have the uniform lower bound (for all $x$, meeting (\ref{L1})):
$$
(1-C r^{-\alpha})^{k-1} \ge e^{-2C k_0 r^{-\alpha}} \ge e^{-2 C} = \tilde C_0.
$$
In addition,
$$
(1+|x|+kr)^{-(d+\alpha)}\ge (1+2|x|)^{-(d+\alpha)}.
$$
Finally, we have
\begin{equation}\label{ee_bis}
a_k(x)\geq C_0(1+|x|)^{-(d+\alpha)}, \quad C_0 = 2^{-(d+\alpha)} \tilde C_0 c_-, \quad \alpha>0.
\end{equation}
Summing up in $k$ and recalling inequality (\ref{L1}), we obtain
$$
G_\lambda(x)\geq C_0(1+|x|)^{-(d+\alpha)}\sum\limits_{k=1}^{k_0(|x|)}
(1+\lambda)^{-k} \ge
$$
$$
C_0(1+|x|)^{-(d+\alpha)}\frac{1 - (1+\lambda)^{-|x|^{\frac\alpha{\alpha+1}}}}{\lambda}
\geq\frac{C_0}{2 \lambda}(1+|x|)^{-(d+\alpha)}.
$$
We used in the last inequality that the function $g(u)=(1+u)^{-\frac{1}{u}}$ is increasing as $u \in (0,1)$, and $g(0)=e^{-1}, \, g(1)=\frac12$.
\end{proof}
Theorem \ref{t_1} is completely proved.

\end{proof}

\vskip8mm

In the next section we study the behavior of the function $G_\lambda(x)$ in the case of exponentially decaying convolution kernel $a(x)$.

\subsection{Exponential tail of  dispersal kernel}

In this section we consider the exponentially decaying dispersal kernels. Namely, we assume that in addition to (\ref{aL1}) - (\ref{a2}) the function $a(x)$ satisfies the following upper bound:
\begin{equation}\label{expbound}
a(x) \ \le \ c \ \exp\{- \delta |x|\}
\end{equation}
with a positive constant $\delta$. We prove here that the function $G_\lambda(x)$ defined by (\ref{Gbis}) decays exponentially.
The proof relies essentially on a probabilistic approach.

\begin{theorem}\label{t_2} There exist positive constants $k = k(\lambda), \ m = m(\lambda)$, such that the function $G_\lambda(x)$ satisfies the following upper bound:
\begin{equation}\label{exp_ineq}
G_\lambda(x) \leq \  k(\lambda) e^{- m (\lambda) |x|}.
\end{equation}
Moreover, $k(\lambda) \to \infty$ and $m(\lambda)= O(\lambda)$, as $\lambda \to 0$.

\end{theorem}

\begin{proof}

We start the proof with the following Lemma.

\begin{lemma}\label{lexp}
Let $X_i$ be i.i.d. random variables taking values in $\mathbb R^d$ with the distribution density $a(x)$ satisfying the upper bound (\ref{expbound}). Consider a unit vector $\theta \in \mathbb R^d$, and denote $S_n = X_1+ \ldots +X_n$. Then there exists a constant $c_1 = c_1(c, \delta)$ such that for all $n \ge 1$ the following estimates hold:
\begin{equation}\label{theta}
P \{ |\theta \cdot S_n | > r \} \le \left\{
\begin{array}{ll}
\displaystyle
e^{- \frac{r^2}{4 c_1 n}},&  0<r \le    \delta  c_1 n \\ [2mm]
e^{- \frac{\delta r}{4}},&  r >    \delta c_1 n
\end{array} \right.
\end{equation}
where $\delta$ is the same as in (\ref{expbound}).
\end{lemma}

\begin{proof}

Estimate (\ref{expbound}) on the density $a(x)$ is isotropic, hence we can take $\theta = e_1$ and consider an 1-d random variable $\xi = \theta \cdot X_1 = X_1^{(1)}$. The distribution of $\xi$ satisfies the estimate similar to (\ref{expbound}) with the same $\delta$ and some constant $\tilde c =  \tilde c(c, \delta)$:
$$
a_\xi(x)  \ \le \ \tilde c e^{- \delta |x|}, \quad x \in \mathbb R^d.
$$
Let  $k \in (0, \frac\delta2]$ be a constant, then  using the Taylor decomposition for the exponent $e^{k|\xi|}$ and estimates on the moments of $\xi$ we get
\begin{equation}\label{xiexp}
{\mathbb E} e^{k|\xi|} \ \le \ e^{c_1 k^2} \quad \mbox{ for } \; 0<k\le \frac\delta2
\end{equation}
with a constant $c_1 = c_1(\tilde c, \delta)$.

Let us take a unit vector $\theta \in \mathbb R^d$ and fix $n \in \mathbb N$. Then
using the Chebyshev inequality, the independence of random variables $X_j$ and (\ref{xiexp}) we get for the 1-d random variable $\theta \cdot S_n$:
$$
P \{ |\theta \cdot S_n | > r \} = P \{ k |\theta \cdot S_n | > k  r \}  \le \min\limits_{0<k \le \frac{\delta}{2}} \frac{(\mathbb{ E} e^{k |\theta \cdot X_1|})^n}{e^{kr}} \le
\min\limits_{0<k \le \frac{\delta}{2}} e^{ c_1 k^2 n - kr} = \exp \big\{  \min\limits_{0<k \le \frac{\delta}{2}} f(k)   \big\}
$$
with $f(k) =  c_1 k^2 n - kr$. We consider two cases. If
$$
k_0 = \mathrm{argmin} f(k) = \frac{r}{2 c_1 n} \le \frac\delta2,
$$
then $\min\limits_{0<k \le \frac{\delta}{2}} f(k) = f(k_0) =  -\frac{r^2}{4 c_1 n}$, and
$$
P \{ |\theta \cdot S_n | > r \} \le e^{- \frac{r^2}{4 c_1 n}}.
$$
If $r>\delta c_1 n$, then $k_0 > \frac\delta2$. Consequently, $\min\limits_{0<k \le \frac{\delta}{2}} f(k) = f(\frac\delta2)$ and
$$
P \{ |\theta \cdot S_n | > r \} \le e^{ \frac{c_1 \delta^2 n}{4} - \frac{\delta r}{2}} = e^{(\frac{c_1 \delta^2 n}{4} - \frac{\delta r}{4}) - \frac{\delta r}{4}} \le  e^{- \frac{\delta r}{4}}.
$$
\end{proof}

For a $d$-dimensional random vector $\eta$ and arbitrary $\varepsilon>0$ one can find a finite collection of unit vectors $\theta_1, \ldots, \theta_N, \; N=N(\varepsilon, d)$ such that
$$
\{ |\eta | > r \}  \subset \bigcup_{i=1}^N \{ |\theta_i \cdot \eta| >(1-\varepsilon) r \}
$$
Then
% and approximate event $\{|\eta|>r \}$ for $d$-dimensional random variable $\eta$ from above by the union $\cup_{i=1}^N \{ |\theta_i \cdot \eta| >(1-\varepsilon) r \}$. For any $\varepsilon>0$ one can find corresponding $N=N(\varepsilon, d)$ with
\begin{equation}\label{N}
P \{ |\eta | > r \} \le N(\varepsilon, d) \ P \{ |\theta_1 \cdot \eta | > (1-\varepsilon )r \}.
\end{equation}
Together with the result of Lemma \ref{lexp} it gives
\begin{equation}\label{P_r}
P \{ | S_n | > r \} \le \left\{
\begin{array}{ll}
\displaystyle
c_2 \ e^{- \frac{r^2}{4 c_1 n}},&  0<r \le    d_1 n \\ [2mm]
c_3 \ e^{- \frac{\delta r}{4}},&  r >    d_1 n
\end{array} \right.
\end{equation}
with constants $c_2, \ c_3, d_1$ which do not depend on $r$ and $n$.

 Next we show that the density $a_n (x)$ satisfies the estimate similar to (\ref{P_r}). Indeed, denote  by $F_n(y)$ the distribution function of $S_n$, then from (\ref{expbound}) it follows
$$
a_{n+1}(x) = \int\limits_{\mathbb R^d} a(x-y) \ dF_n (y) \le c \int\limits_{\mathbb R^d} e^{-\delta |x-y|} \ dF_n (y) =
$$
$$
c \int\limits_{|y-x| \le \frac12 |x| } e^{-\delta |x-y|} \ dF_n (y) + c \int\limits_{|y-x|>\frac12 |x|} e^{-\delta |x-y|} \ dF_n (y) \le
c P \big\{ |S_n|\ge \frac12 |x| \big\} + c e^{-\frac12 \delta |x|}.
$$
Together with (\ref{P_r}) this yields for all $n \ge 1$:
\begin{equation}\label{Pn_r}
a_n(x) \le \left\{
\begin{array}{ll}
\displaystyle
\tilde c_1 \ e^{- l_1 \frac{|x|^2}{n}},&  |x| \le B n \\ [2mm]
\tilde c_2 \ e^{- l_2 |x|},&  |x| >  B n
\end{array} \right.
\end{equation}
with constants $\tilde c_1, \tilde c_2, l_1, l_2, B$ which do not depend on $|x|$ and $n$.
Let us estimate now $G_\lambda(x)$ from above using (\ref{G}) and (\ref{Pn_r}):
$$
G_\lambda (x) \ = \ \sum_{n=1}^{\left[ \frac{|x|}{B} \right]} \frac{a_n(x)}{(1+\lambda)^n} +  \sum_{n > \left[ \frac{|x|}{B} \right]} \frac{a_n(x)}{(1+\lambda)^n} \  \le
$$
$$
\frac{\tilde c_2}{\lambda} e^{-l_2|x|} + \tilde c_1  \sum_{n \ge \left[ \frac{|x|}{B} \right]} \frac{1}{(1+\lambda)^n} \ \le \ \frac{\tilde c_3}{\lambda} \left( e^{-l_2|x|} + e^{-l_3(\lambda)|x|} \right) \ \le \ k(\lambda) e^{-m(\lambda)|x|}
$$
with $k(\lambda) = \frac{2 \tilde c_3}{\lambda}$, and $l_3(\lambda) = \frac1B \ln(1+\lambda)$, $m(\lambda) = \min \{ l_2,\ l_3(\lambda)\}$. Consequently, if $\lambda$ is small enough, then $m(\lambda) = l_3(\lambda) = \frac{\lambda}{B}(1+o(1))$.

\end{proof}

\begin{remark}
If we assume that there exist two constants $c_1, c_2$, such that
$$
c_1 e^{-\delta |x|} \le a(x) \le c_2 e^{-\delta |x|} ,\quad x  \in \mathbb R^d,
$$
then
$$
k_1 (\lambda) e^{- \delta |x|} \leq G_\lambda(x) \leq \  k_2(\lambda) e^{- m (\lambda) |x|}
$$
with positive constants $k_1(\lambda), k_2(\lambda), m(\lambda)$.
\end{remark}

In the one-dimensional case, $d=1$, the constants $m(\lambda)$ in Theorem \ref{t_2} can be found more precisely.
We assume that the function $a(x)$ meets the following asymptotics at infinity:
\begin{equation}\label{aexp}
\lim_{|x| \to \infty} \frac{\ln a(x)}{|x|} = -c
\end{equation}
with a constant $c>0$. Then the Fourier transform $\tilde a(p)$ is an analytic function in the strip $|Im \ p|<c$.
Assume additionally some smoothness of the function $a(x)$ guaranteeing that
\begin{equation}\label{a''}
\tilde a_\kappa (v) \equiv \tilde a(i\kappa+v) = \int a(x) e^{\kappa x} e^{-ivx} dx \in L^1(\mathbb R) \quad \mbox{ for any } \; 0 \leq \kappa<c.
\end{equation}
In particular, (\ref{a''}) is valid if  $a(x) \in C^2(\mathbb R)$ and $\left(a(x)e^{\kappa x}\right)', \left(a(x)e^{\kappa x}\right)'' \in L^1(\mathbb R)$ for any $0 \leq \kappa<c$.

We are interested in the asymptotcis at infinity of $G_\lambda(x)$ introduced in (\ref{G}).
Let us consider solutions  $\hat p = \hat p(\lambda)$  of the equation
\begin{equation}\label{akappa}
\tilde a(\hat p) = 1+\lambda.
\end{equation}
It is easy to see, that if there exist solutions of (\ref{akappa}), then there exists at least a pure imaginary solution $\hat p = iq$, and  $\beta>q$ for any other solutions $\hat p = \alpha +i\beta$ of (\ref{akappa}). Indeed, considering the properties of $a(x)$ we conclude that $\tilde a(iq)$ is a continuous positive increasing function of $q$, and $\tilde a(iq) > |\tilde a(iq+\alpha)|$.

\begin{theorem}\label{l2}
Assume that $a(x)$ satisfies conditions (\ref{aL1}) - (\ref{a2}) and (\ref{aexp}), and let $\lambda>0$ be a positive constant. There exist constants $C(\lambda), C_{-}(\lambda), C_{+}(\lambda)$ such that \\
1) if there exists a pure imaginary solution $\hat p(\lambda) = iq(\lambda) $ of equation (\ref{akappa})
with $q (\lambda) <c$, then $G_\lambda(x)$ has the asymptotics
\begin{equation}\label{Exp1}
G_\lambda(x) = e^{-q(\lambda) |x|}(C(\lambda) + o(1)), \quad |x| \to \infty;
\end{equation}
2) otherwise we have the following two-sided bound:
\begin{equation}\label{Exp2}
C_{-}(\lambda) e^{-(c+ \varepsilon) |x|} \le G_\lambda(x) \le C_+(\lambda) e^{-(c-\varepsilon)|x|}
\end{equation}
with any $\varepsilon>0$.
\end{theorem}

\begin{remark}
Notice that if
\begin{equation}\label{remark}
\lim_{q\to c-} \int_{\mathbb R} a(x) e^{q x} dx = + \infty,
\end{equation}
then the equation (\ref{akappa}) has a pure imaginary solution $\hat p(\lambda) = iq(\lambda) $ for any $\lambda>0$. In this case the asymptotics of $G_\lambda(x)$ is given by (\ref{Exp1}).

If the limit (\ref{remark}) is finite, then depending on the value of $\lambda$ both (\ref{Exp1}) and (\ref{Exp2}) can realize.

\end{remark}

\begin{proof}

The proof relies on the analyticity of $\tilde a$ in an appropriate strip. We use the representation for $G_\lambda(x)$ in terms of  $\tilde a$:
\begin{equation}\label{GF}
G_\lambda(x) = \frac{1}{2 \pi} \int_{\mathbb R} \frac{\tilde a(p)}{\lambda+1 - \tilde a(p)} e^{i(p,x)} dp.
\end{equation}
Let $x>0$, and construct a rectangular closed contour containing a segment of a real line $[-K,K]$ a parallel segment $[K + i\kappa, -K + i \kappa]$ and two segments $I_1(K),\ I_2(K)$ parallel to the imaginary axis: $I_1(K) = [K, K+i\kappa], \ I_2 (K) = [-K+i\kappa, -K]$, where $0 \leq \kappa<c$. In the first case when $  q < c$ we can take $\kappa: q< \kappa<c$ with no other solutions of (\ref{akappa}) in the strip $0< \mathcal{I}\mathrm{m} \ p<\kappa$. Then we have:
$$
\int_{-K}^K  \frac{\tilde a(p)}{\lambda+1 - \tilde a(p)} e^{i(p,x)} dp + \int_{I_1(K) \cup I_2(K)}  \frac{\tilde a(p)}{\lambda+1 - \tilde a(p)} e^{i(p,x)} dp +
$$
$$
\int_{K}^{-K}  \frac{\tilde a(i\kappa+u)}{\lambda+1 - \tilde a(i\kappa+u)} e^{i(i\kappa+u)x} du \ = \ 2 \pi i \ res_{q}  \left( \frac{\tilde a(p)}{\lambda+1 - \tilde a(p)} e^{i(p,x)} \right),
$$
and
\begin{equation}\label{22}
\begin{array}{c}
\displaystyle
G_\lambda(x) =  2 \pi i \ res_{q}  \left( \frac{\tilde a(p)}{\lambda+1 - \tilde a(p)} e^{i(p,x)} \right) + \int_{\mathbb R} \frac{\tilde a(i\kappa+u)}{\lambda+1 - \tilde a(i\kappa+u)} e^{-\kappa x + i u x} du -
\\[4mm]
\displaystyle
\lim_{K \to \infty} \int_{I_1(K) \cup I_2(K)}  \frac{\tilde a(p)}{\lambda+1 - \tilde a(p)} e^{i(p,x)} dp.
\end{array}
\end{equation}

Let us prove first that the limit in (\ref{22}) is equal to 0. Since the function $a(x) e^{ux} \le a(x) e^{\kappa x} \le B e^{-\delta |x|}$ with some $B>0$ and $\delta>0$ uniformly in $u\in [0,\kappa]$, then
$$
\int_{\mathbb R} a(x) e^{ux} dx < \infty \quad \mbox{ for all } \;   u\in [0,\kappa].
$$
Consequently,
$$
\tilde a(K+iu) = \int a(x) e^{-i (K+iu) x} dx =  \int a(x) e^{ux} e^{-iKx} dx \to 0 \quad \mbox{ as } \; K \to \infty
$$
uniformly in  $u\in [0,\kappa]$. Analogously, the function $\tilde \Phi (p) = \frac{\tilde a(p)}{\lambda+1 - \tilde a(p)}$ is uniformly bounded on $I_1(K)\cup I_2(K)$  and
$|  \tilde \Phi (K+iu)| \to 0$ as $K \to \infty$ uniformly over $u \in [0,\kappa]$. Since $|e^{i(\pm K+iu, x)}|=e^{-u x}\le 1$, then we obtain that
$$
\int_{I_1(K) \cup I_2(K)}  \frac{\tilde a(p)}{\lambda+1 - \tilde a(p)} e^{i(p,x)} dp \to 0 \quad \mbox{ as } \; K \to \infty.
$$

Let us consider the upper segment $[-K + i\kappa, K + i \kappa]$ of the contour. As above we have that the function  $ \tilde a ( i\kappa+u)$ is uniformly bounded for all $u \in [-K,K]$, and since $ a(x) e^{\kappa x} \in L^1 (\mathbb R)$ we get
$$
\tilde a(i\kappa+u) = \int a(x) e^{\kappa x} e^{-i u x} dx \to 0 \quad \mbox{ as } \; |u| \to \infty.
$$
Thus $ \tilde \Phi (i\kappa+u)$ is uniformly bounded for all $u \in [-K, K]$, and (\ref{a''}) implies that
$$
\int_{\mathbb R}  |\tilde \Phi (i \kappa+u)| du < \infty.
$$
Consequently, we have for the integral in (\ref{22}):
\begin{equation}\label{I}
\int_{\mathbb R}  \frac{\tilde a(i\kappa+u)}{\lambda+1 - \tilde a(i\kappa+u)} e^{-\kappa x + iu x} du \ = \ O(e^{-\kappa x}),
\end{equation}
and the main contribution to the asymptotics for (\ref{GF}) comes from the residue at $iq$, and we get
$$
G_\lambda (x) \sim 2 \pi i \ \mathop{\mathrm{res}}\limits_{iq}  \left( \frac{\tilde a(p)}{\lambda+1 - \tilde a(p)} e^{i(p,x)} \right) =  e^{-q(\lambda) x}(C(\lambda) + o(1)), \quad x>0.
$$

In the second case the above sum over the closed contour is equal to 0, and the main contribution to the asymptotics of $G_\lambda(x)$ comes from the term (\ref{I}) with any $\kappa<c$. Consequently, in this case we can conclude the following upper bound (as $x>0$):
$$
G_\lambda(x) \le C_+(\lambda)  e^{-(c-\varepsilon) x} \quad \mbox{ for any } \; \varepsilon>0.
$$

The lower estimate on $G_\lambda(x)$ immediately follows from representation (\ref{G}) if we take in (\ref{G}) the first term:
$$
G_\lambda(x) > \frac12 a(x) \ge C_-(\lambda) e^{-(c + \varepsilon) x} \quad \mbox{ with any } \ \varepsilon>0.
$$

The case when $x<0$ can be considered in the same way using a rectangular closed contour in the negative imaginary semi-plane.

\end{proof}

\begin{remark}
It follows from the last Theorem that for small $\lambda$ we have $q(\lambda)=O(\sqrt{\lambda})$. In the one dimensional case this  improves the general bound \eqref{exp_ineq} where $m(\lambda)=O(\lambda)$.
\end{remark}

\section{Applications}

In this section we present some applications of the results obtained in the previous section.

\subsection{Asymptotic of the principal eigenfunction $\psi_\lambda$}

We consider the operator
\begin{equation}\label{2}
L u(x) \ = \  L_0 u(x) \ + \ V(x) u(x), \quad u(x) \in  L^2({\mathbb R}^d),
\end{equation}
with $L_0$ defined by (\ref{3}). For the potential $V$, we assume that
\begin{equation}\label{3bis}
0 \le V \le 1, \quad  V(x) \in C_0({\mathbb R}^d).
\end{equation}
The operator $L$ is a bounded self-adjoint operator in $ L^2 ({\mathbb R}^d)$.
The equation for the principal eigenfunction (the ground state) $\psi_\lambda$ of the operator $L$
\begin{equation}\label{L0}
(L_0 + V - \lambda) \psi_\lambda \ = \ 0, \quad \lambda>0,
\end{equation}
can be rewritten in the following way
$$
(1+\lambda)\psi_\lambda(x) -\mathcal{S}_a \psi_\lambda(x)=V(x)\psi_\lambda(x)=:F(x).
$$
where $\mathcal{S}_a$ stands for the convolution operator. After proper rearrangements
this yields
\begin{equation}\label{GS}
\psi_\lambda(x) = (1+\lambda)^{-1}\Big(F(x)+\sum\limits_{k=1}^\infty\frac{(\mathcal{S}_a^k F)(x)}{(1+\lambda)^k}\Big)= (1+\lambda)^{-1}\Big(F(x)+\int G_\lambda(x-y)F(y)dy \Big),
\end{equation}
where $G_\lambda(x)$ is defined by (\ref{Gbis}).

%\begin{equation}\label{R}
%Q_{\lambda} \psi (x) \  = \  \psi (x),
%\end{equation}
%where $Q_{\lambda}$ is a compact positive operator defined by
%\begin{equation}\label{Q}
%Q_{\lambda} = \frac{1}{\lambda+1} W_\lambda^{-1} A_\lambda V.
%\end{equation}
%Here
%\begin{equation}\label{a}
%A_{\lambda} = (\lambda+1) (\lambda - L_0)^{-1} - 1 =
%\sum_{n=1}^{\infty} \frac{a^{\ast n}}{ (\lambda+1)^n},
%\end{equation}
%and $W_\lambda$ is an operator of multiplication by the function
%$$
%W_\lambda = 1-\frac{V}{\lambda+1}.
%$$
%$A_\lambda$ is a bounded convolution operator when $\lambda \in \mathcal{D}$ with the following kernel
%\begin{equation}\label{G}
%G_{\lambda}(x-y) = \frac{1}{(2 \pi)^d} \int_{\mathbb{R}^d} e^{-i(p, x-y)} \frac{\tilde a (p) \ d p}{\lambda + 1- \tilde a(p)} \ = \
%\left(  \sum_{n=1}^{\infty} \frac{a^{\ast n}}{ (\lambda+1)^n} \right) (x-y).
%\end{equation}

%Consequently,  $Q_{\lambda}$ is an analytic operator-valued function, such that  $Q_{\lambda}$ is a compact positive operator for any %$\lambda \in \mathcal{D}  = \{ \lambda \in \mathbb{C} \ | \ \ Re \lambda>0 \} $.

We remind below of the spectral properties of operator $L$ and of the conditions ensuring the existence of the principle eigenfunction $\psi_\lambda(x)$. Notice that the operator $L$ has these properties for any kernel $a(x)$ that meets conditions (\ref{aL1}), independently on the behaviour of the tail of $a(x)$. The following statements have been proved in \cite{KMPZh}.

%Since the operator $L$ is bounded and $Q_\lambda$ is the compact positive operator for any $\lambda \in \mathcal{D}$, then by the %Krein-Rutman theorem it follows that if there exists the maximal positive eigenvalue of $L$ and the corresponding eigenfunction %$\psi(x)>0$ (the ground state of $L$) is positive. The uniqueness of the ground state $\psi_\lambda(x) >0$ of the operator $L$ in the %space $L^2({\mathbb R}^d)$ follows from the positivity improving property of the semigroups $e^{tL_0}$ and $e^{tL}$, see e.g. \cite{RS4}, %Theorem XIII.44. The last semigroup is positivity improving due to the Feynman-Kac formula.

%\begin{lemma}
%\label{Lemma 2} {\rm \cite{KMPZh}} If the potential $V(x) \in C_0({\mathbb{R}}^d)$ and the operator $L$ %has a ground state $\psi_\lambda(x) \in  L^2({\mathbb{R}}^d)$, then the ground state  $\psi_\lambda (x) %\in C_b({\mathbb{R}}^d) \cap  L^2({\mathbb{R}}^d)$ and  $\psi_\lambda (x) \to 0$ as $|x| \to \infty$.
%\end{lemma}

\begin{theorem}\label{Theorem 1}{\rm \cite{KMPZh}}
\begin{itemize}
\item
The operator $L$ has only discrete spectrum in the half-plane $\mathcal{D} = \{ \lambda \in \mathbb{C} \ | \ \ Re \lambda>0 \}$.
\item
Assume that $V(x) \equiv 1$ on some open set in $\mathbb R^d$. Then the ground state of $L$ exists.
\item For any $\delta >0$ there is $\varepsilon>0$ such that for any potential $V(x)$ that satisfied
the inequality $V(x) \ge 1- \varepsilon$ on a ball of radius $\delta$,  the ground state of $L$ exists.
\item
If $V(x)\geq \beta$, $\beta\in (0,1)$, in a large enough ball (depending on $\beta$) then the ground state exists.
\item Let $d=1,2$, and assume that $\int_{\mathbb R^d} |x|^2 a(x) dx < \infty $. Then for any $V(x)\not \equiv 0$
 the ground state of  $L$ exists.
\end{itemize}
\end{theorem}

We proceed with studying the asymptotic behavior of the function $\psi_\lambda(x)$ as $|x| \to \infty$. %depending on the asymptotics of the tail of the function $a(x)$.

\begin{theorem}\label{T1_1}  Let $V\in C_0(\mathbb R^d)$, and assume that $V(x)\not \equiv 0$. If
\begin{equation}\label{bou_a}
c_-(1+|x|)^{-(d+\alpha)}\leq a(x)\leq c_+(1+|x|)^{-(d+\alpha)},
\end{equation}
then the principal eigenfunction satisfies the estimates
\begin{equation}\label{main_ineq}
\tilde c_-(\lambda)(1+|x|)^{-(d+\alpha)}\leq \psi_\lambda(x)\leq \tilde c_+(\lambda)(1+|x|)^{-(d+\alpha)}
\end{equation}
with $0<\tilde c_-(\lambda)\leq \tilde c_+(\lambda)$. \\ \\
If  bound (\ref{expbound}) holds, then for $\psi_\lambda(x)$  the following estimate holds:
\begin{equation}\label{E0bisbis}
\psi_\lambda (x) \le K(\lambda)e^{-m(\lambda) |x|}. \\ \\
\end{equation}
Let $d=1$, and assume that condition (\ref{aexp})
is fulfilled. Then
\begin{itemize}
\item if there exists a solution $\hat p(\lambda) = iq(\lambda) $ of equation (\ref{akappa})
with $q(\lambda) <c$, then
\begin{equation}\label{E1bis}
\psi_\lambda (x) = e^{-q(\lambda) |x|}(C(\lambda) + o(1)), \quad |x| \to \infty,
\end{equation}
\item if the solution of (\ref{akappa}) does not exists for $q<c$, then
\begin{equation}\label{E2bis}
C_-(\lambda) e^{-(c+ \varepsilon) |x|} \le \psi_\lambda(x) \le C_+(\lambda) e^{-(c-\varepsilon)|x|}
\end{equation}
with any $\varepsilon>0$.
\end{itemize}

\end{theorem}

\begin{proof}
Since function  $F(x) = V(x) \psi_\lambda(x)$ has a compact support and is positive, then relations (\ref{main_ineq}), \eqref{E1bis} and \eqref{E2bis}  follow from representation (\ref{GS}) and estimates (\ref{main_ineq1}), (\ref{Exp1})-(\ref{Exp2}) on the function $G_\lambda(x)$.
\end{proof}

\subsection{Application to the solution of the Cauchy problem}

Let us consider the following Cauchy problem 
\begin{equation}\label{ACP}
\frac{\partial u}{\partial t} = L_0 u(x,t) - m u(x,t) + f(x), \quad u(x,0) = 0,
\end{equation}
where $m>0$ is a constant, $f(x) \geq 0, \; f \not \equiv 0$, $f\in L^1(\mathbb R^d)\cap C_b(\mathbb R^d)$. Equation \eqref{ACP} describes the evolution of the population density in a spacial contact model with mortality rate $ 1+ m$ and space inhomogeneous source term $f(x)$.  In the absence of
the source term, the density of the population goes to zero as $t\to \infty$ (see  \cite{KKP}).  The question of interest is how a flow  coming from the source into
the population may change the asymptotics of the density.

Since $S(t) = e^{t(L_0 - m)}$ is a contraction semigroup both in $L^2(\mathbb R^d)$ and $C_b(\mathbb R^d)$ spaces,
the solution $u(x,t)$ of problem (\ref{ACP}) converges to the corresponding stationary solution $\hat u(x)$:
$$
\| u(x,t) - \hat u(x) \| \to 0, \quad t \to \infty.
$$
The convergence takes place both in $L^2(\mathbb R^d)$ and in $C_b(\mathbb R^d)$ norms.
Concerning the stationary solution  $\hat u(x)$ let us observe that according to (\ref{GS}) it can be expressed as
\begin{equation}\label{hatu}
\hat u(x) = (m-L_0)^{-1} f(x) = \frac{1}{1+m} \Big( f(x) + \int G_m(x-y) f(y) dy \Big).
\end{equation}
It turns out that the behaviour of $u(x,t)$ and $\hat u(x)$ depends on both the tail of the convolution kernels and of the source term.

As above we consider separately the cases of polynomial and  exponential
tails of $a$ and $f$.
Assume first that the function $a(x)$ satisfies inequalities (\ref{bou_a1}) and that the function $f(x)$ satisfies analogous inequalities with some $\alpha_1$.
Then, due to \eqref{main_ineq1},  $\hat u(x)$ admits the following bounds
 \begin{equation}\label{bou_uhat}
c_-(1+|x|)^{-(d+\tilde\alpha)} \leq \hat u(x) \leq c_+(1+|x|)^{-(d+\tilde\alpha)},
\end{equation}
with $\tilde\alpha= \min(\alpha_1, \alpha)$.
Moreover, $u(x,t)$ at any time $t>0$ admits the inequality
\begin{equation}\label{hatubis}
0 \leq u(x,t) \leq \hat u(x).
\end{equation}
To prove this inequality we define $v(x,t) = u(x,t) - \hat u(x)$. Then $v(x,0) = - \hat u(x) < 0$. Since $e^{t(L_0 - m)}$ is the positivity improving semigroup, then $v(x,t) <0$ for all $t>0$. Therefore, $u(x,t)< \hat u(x)$ for all $t>0$. The lower bound in \eqref{hatubis} follows from Duhamel's formula.

In the case of exponentially decaying $a$ and $f$ both $\hat u$ and $u$ decay exponentially.
This can be justified in the same way as above.

If $f(x) \in C_0(\mathbb R^d)$, then as a direct consequence of \eqref{hatu}, the stationary
solution  $\hat u(x)$ satisfies the same tail estimates as the dispersal kernel.

In conclusion we provide some comments on possible interpretation of the above results.
\begin{itemize}
\item Consider the contact model in continuum (see \cite{KKP}) describing an infection spreading
process in a society with the recovering intensity $1+m$. For $m>0$ this intensity is sufficient
to make the density of infected population degenerate. Suppose that the source of the infection
is localized, for instance $f(x)=\lambda  \mathbf{1}_B(x)$,  $\lambda >0$, with a bounded (small) set $B$ where infected individual appear (from outside).
Then we can estimate the density of infected population on a distance from the infection source $B$ which essentially depends on the infection spreading rate. Even very small region $B$ may produce a drastic effect!
\item The same process we have in the information spreading in the society.
Even if you will have
very strong  real information delivering rate ($m>>1$) the influence of  a constant mass-media flow of wrong information may be determining for the opinion formation.
Especially due to long range  spreading possibilities (TV). 
\item Free Kawasaki dynamics of continuous particle system (see \cite{KKOSS}) can be used for modeling a
pollution spreading process.  Equation \eqref{ACP} with $f=0$ and some $u_0$ describes the evolution of the pollution density in the presence of a self cleaning ability of the environment with intensity $m>0$. In the general case the function $f\geq 0$ represents the density of the pollution source. The solution $u(x,t)$ is the density
of the pollution after time $t$, and $\hat u(x)$ is the large time limit of this density.
The estimates \eqref{bou_uhat} reflect the fact that even for a source localized in a small area,
for instance for $f=\lambda  \mathbf{1}_B(x)$ with a small ball $B$, the strong pollution spreading can be observed at long distances from the source.
\end{itemize}

%%%%%%%%%%%%%%%%%%%%%%

\end{document}